\documentclass[letterpaper, 10 pt, conference]{ieeeconf}  

\IEEEoverridecommandlockouts                              

\usepackage{graphics} 
\usepackage{epsfig} 
\usepackage{mathrsfs}
\usepackage{times} 
\usepackage{amsmath} 
\usepackage{amssymb}  
\usepackage{cite}
\usepackage{bm}
\usepackage{color}
\usepackage{xcolor}
\usepackage{framed}
\usepackage{dsfont}
\usepackage{arydshln}

\usepackage{algorithm}
\usepackage{algorithmicx}
\usepackage{algpseudocode}
\algrenewcommand{\algorithmiccomment}[1]{\hskip3em\% #1}

\usepackage[small]{caption}
\usepackage[subrefformat=parens]{subcaption}
\newtheorem{lemma}{Lemma}
\newtheorem{theorem}{Theorem}
\newtheorem{corollary}{Corollary}

\newtheorem{proposition}{Proposition}

\newcommand{\A}{\mathcal{A}}

\newcommand{\cC}{\mathcal{C}}

\newcommand{\f}{\mathsf{f}}

\newcommand{\I}{\mathcal{I}}

\newcommand{\cP}{\mathcal{P}}

\newcommand{\p}{\mathsf{p}}
\newcommand{\R}{\mathbb{R}}

\newcommand{\cS}{\mathcal{S}}

\newcommand{\U}{\mathcal{U}}

\newcommand{\W}{\mathcal{W}}

\newcommand{\Z}{\mathbb{Z}}

\newcommand{\col}{\mathrm{col}}

\newcommand{\im}{\mathrm{im\hspace{0.2ex}}}

\newcommand{\minimize}{\mathop{\rm minimize}\limits}

\title{\LARGE \bf
Data-Driven Control from Poisoned Data:\\Fundamental Limitations and Secure DeePC
}

\author{Takumi Shinohara, Henrik Sandberg, and Karl Henrik Johansson \\
\thanks{This work was supported in part by the Knut and Alice Wallenberg Foundation Wallenberg Scholar Grant, Swedish Research Council (Project 2023-04770), Swedish Research Council Distinguished Professor Grant (Project 2017-01078), Swedish Civil Defence and Resilience Agency (Project MAD-VAMCHS), and VINNOVA project ``Control-computing-communication co-design for autonomous industry (3C4AI)'' (Project 2025-01119).}
\thanks{The authors are with the Department of Decision and Control Systems, KTH Royal Institute of Technology, and also with Digital Futures, 100 44 Stockholm, Sweden. (e-mail: tashin@kth.se, hsan@kth.se, kallej@kth.se).}}

\begin{document}

\maketitle
\thispagestyle{empty}
\pagestyle{empty}

\begin{abstract}
We study a data-driven control problem in the presence of arbitrary data poisoning attacks. We assume that a subset of offline output data is stored in unprotected locations and may be poisoned by an adversary. We first establish fundamental limitations for data-driven control arising from such poisoned data: poisoning attacks are not detected/identified from the dataset alone; unprotected data are non-informative for controller design with worst-case guarantees; and hard constraints on unprotected outputs are not certifiable. Motivated by these limitations and the data-enabled predictive control (DeePC) technique, we propose \textit{Secure DeePC}, a data-driven control algorithm that is resilient against poisoning attacks. It first runs output-truncated DeePC using only the protected dataset until the online input becomes persistently exciting. It then uses online measurements to reconstruct the partial offline dataset, and finally returns to full-output DeePC. Secure DeePC achieves MPC-equivalent performance in finite time almost surely under certain conditions. Simulation results illustrate the efficacy of the proposed framework against poisoning attacks.
\end{abstract}


\section{Introduction}
\label{section:introduction}
Recent advances in data-driven methods have spurred significant interest in direct data-driven control \cite{2019ECCDorfler,2023TACDorfler,2021TACAllgower,2020TACTesi}, where controllers are synthesized directly from measured offline input-output (I/O) data of systems with unknown dynamics, thereby avoiding explicit model identification.
Among data-driven control schemes, data-enabled predictive control (DeePC) \cite{2019ECCDorfler,2023TACDorfler} has emerged as a prominent approach that leverages Willems' Fundamental Lemma to formulate a model predictive control (MPC) problem directly in terms of Hankel matrices constructed from data.

Most of these works implicitly assume that the collected data are genuine and have not been corrupted.
However, adversaries can manipulate the data by exploiting vulnerabilities in data storage or communication channels.
Indeed, some studies \cite{2021ACCProutiere,2023ACCTopcu,2023CDCSasahara} have shown that poisoning attacks can severely degrade performance and even destabilize data-driven controllers.
From the defender's viewpoint, \cite{2025LCSSShi,2025AutomaticaZhang} develop data-driven secure controllers that operate under intermittent loss of offline data induced by denial-of-service (DoS) attacks.
Yet, secure data-driven control against \textit{arbitrary data poisoning} has received limited attention.
Unlike DoS-induced packet losses, in which missing samples are explicitly identifiable and can be treated as erasure data, sophisticated poisoning attacks provide neither explicit evidence of an attack nor reliable indicators of which channels/samples are corrupted.
Such adversarial characteristics make data-driven control significantly more challenging.

\begin{figure}[t]
\begin{center}
\includegraphics[width=\linewidth]{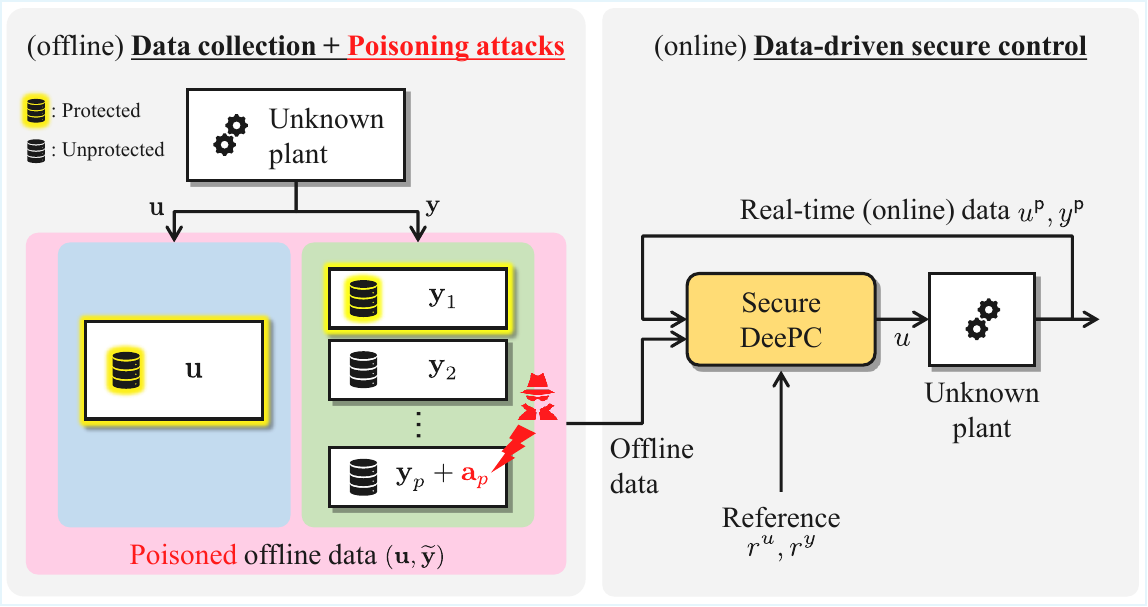}
\vspace{-5.5mm}
\caption{An overview of data poisoning attacks and Secure DeePC.}
\label{fig:illustration}
\vspace{-6.5mm}
\end{center}
\end{figure}

In this paper, we address the data-driven control problem in the presence of arbitrary data poisoning attacks.
A schematic overview of our formulation is shown in Fig.~\ref{fig:illustration}, where the input data and a subset of output data are stored in protected locations, while the remaining outputs are stored in unprotected locations and may be arbitrarily corrupted.
We place no restrictions on the attack sequence against unprotected data.
Our primary contributions are threefold:
\begin{enumerate}
	\item We establish fundamental limitations for data-driven control from arbitrarily poisoned data: poisoning attacks cannot be detected/identified from the dataset alone, unprotected data are non-informative for data-driven controller design with worst-case guarantees, and safety constraints on unprotected outputs are not fundamentally guaranteed.
	\item We show that, if the protected outputs render the system observable and the online input becomes persistently exciting, then the attack-free offline data in a certain interval can be exactly reconstructed from protected offline data and clean online measurements.
	\item Building on the fundamental limitations, we propose \textit{Secure DeePC}, a two-phase architecture that uses output-truncated DeePC before reconstruction and full-output DeePC after reconstruction, thereby recovering MPC-equivalent operation after a finite transient time.
\end{enumerate}

The remainder of this paper is organized as follows:
Section~\ref{section:problem} presents the system model, attack model, data model, and a brief overview of DeePC.
In Section~\ref{section:limitation}, we show the fundamental limitations for establishing data-driven controllers using poisoned data.
In Section~\ref{section:SecureDeePC}, we develop Secure DeePC.
In Section~\ref{section:Simulation}, we show the simulation results using an inverted pendulum model to illustrate the proposed framework.
We conclude the paper in Section~\ref{section:Conclusion}.

\subsubsection*{Notation}
The symbols $ \R $, $ \R^n$, $ \mathbb{Z}$, and $ \Z^+ $ denote the set of real numbers, $ n $-dimensional Euclidean space, integers, and non-negative integers, respectively.
For vectors $ v_1,\ldots,v_k $, define $ \col(v_1,\ldots,v_k) \triangleq [v_1^\top,\ldots,v_k^\top]^\top $.
The notation $ |\I| $ is used to denote the cardinality of a set $ \I $.
For a vector $ x $ and a positive definite matrix $ P $, we write $ \|x\|_P = \sqrt{x^\top P x} $.
For $ a, b \in \Z $ with $ a \leq b $, we define the integer interval $ [a,b] \triangleq \{c \in \Z:a \leq c \leq b \} $. 
Given $ v : \mathbb{Z} \rightarrow \R^n $ and the interval $ [i,j] $, we define a block vector $ v^{[i,j]} $ as $ v^{[i,j]} \triangleq \col(v(i),v(i+1),\ldots,v(j-1),v(j)) $.
Let $ q $ be a positive integer such that $ q \leq j - i + 1 $ and define the block Hankel matrix of depth $ q $, associated with $ v^{[i,j]} $, as
\begin{align*}
\mathscr{H}_q\!\left(\!v^{[i,j]}\!\right) \!\triangleq \!\left[\!\!
\begin{array}{cccc}
v(i) & v(i\!+\!1) & \!\cdots \!& v(j\!-\!q\!+\!1) \\
v(i\!+\!1) & v(i\!+\!2) &\! \cdots \!& v(j\!-\!q\!+\!2) \\
\vdots & \vdots & \!\ddots \!& \vdots \\
v(i\!+\!q\!-\!1) & v(i\!+\!q) & \!\cdots\! & v(j)
\end{array}\!\!\right]\!.
\end{align*}
Note that the subscript $ q $ refers to the number of block rows of the Hankel matrix.
Then, $ v^{[i,j]} $ is said to be \textit{persistently exciting (PE) of order $ q $} if the Hankel matrix $ \mathscr{H}_q(v^{[i,j]}) $ has full row rank \cite{2005SCLWillems}.

\section{Problem Statement}
\label{section:problem}
In this section, we introduce the system model, attack model, and data representation.
We then briefly outline the conventional DeePC algorithm.

\subsection{System Description}
Consider the discrete-time system given by
\begin{align}
\label{eq:system_model}
x(k+1) = \bar A x(k) + \bar B u(k),~~~y(k) = \bar C x(k),
\end{align}
where $ x(k) \in \R^n $ denotes the unknown system state, $ u(k) \in \R^m $ the control input, and $ y(k) \in \R^p $ the system output.
We assume a minimal realization $ (\bar A, \bar B, \bar C) $, i.e., $ (\bar A,\bar B) $ is controllable and $ (\bar A,\bar C) $ is observable.
For notational convenience, define $\mathcal{P} \triangleq \{1,\ldots,p\}$ as the index set of the sensors.
Given a desired reference trajectory $ (r^u, r^y)$, an input constraint set $ \cC_u \subseteq \R^m $, and an output constraint set $ \cC_y \subseteq \R^p $, our goal is to choose inputs so that the system tracks the reference while satisfying the constraints and minimizing a performance cost function.

\subsection{Data Collection and Poisoning Attack}
We collect I/O data during an offline operation of the system (\ref{eq:system_model}) over the time interval $[0, N-1]$, where $ N $ is a sufficiently large integer.
Let $ \textbf{u} \triangleq \col(\textbf{u}(0),\ldots,\textbf{u}(N-1)) \in \R^{mN} $ denote a sequence of $ N $ inputs applied to (\ref{eq:system_model}) and $ \textbf{y} \triangleq \col(\textbf{y}(0),\ldots,\textbf{y}(N-1)) \in \R^{pN} $ denote the corresponding outputs.
The boldface variables (e.g., $ \textbf{u} $, $ \textbf{y} $) indicate data in an offline operation.

The collected data are stored in storage locations for later use by a data-driven controller.
The input data $ \textbf{u} $ are stored in protected storage.
The output data from each sensor are stored in the corresponding storage locations; that is, we have $ p $ output storage locations, and the $ i $th storage location stores the $ i $th output sequence $ \textbf{y}_i \triangleq \col(\textbf{y}_i(0),\ldots,\textbf{y}_i(N-1))\in \R^N $.
Assume that some output storage locations are protected\footnote{The partial protection assumption is prevalent in security-oriented estimation and control, as it is often prohibitively expensive to protect all sensing, communication, and data resources. This protection can be accomplished through encryption and authentication mechanisms, as well as by enhancing physical security.}, whereas other locations are unprotected, and an adversary may corrupt the stored output data.
Denote the sets of protected and unprotected storage locations by $ \cS \subseteq \cP$ and $ \mathcal{U}\subseteq \cP $, respectively.
Note that $ \cS \cap \mathcal{U} = \emptyset $ and $ \cS \cup \mathcal{U} = \cP $.
Without loss of generality, assume that the first $ s_p $ output channels are protected, i.e., $ \cS \triangleq \{1,\ldots,s_p\} $ and $ \U \triangleq \{s_p+1,\ldots,p\} $, where $ s_p $ may be zero.
The system operator knows which storage locations are protected/unprotected.

For $ i \in \mathcal{U} $, we denote the attack sequence against the $ i $th measurement data in unprotected storage by $ \textbf{a}_i \triangleq \col(\textbf{a}_i(0),\ldots,\textbf{a}_i(N-1))\in \R^N $,
where $ \textbf{a}_i(k) \in \R $ denotes the poisoning attack against the $ i $th sensor data at time $ k $ designed by a malicious attacker.
The attacker is assumed to have complete knowledge of the system's state, I/O data, system model, and storage information, and can generate an attack sequence with arbitrary statistical properties, magnitude, and temporal correlations.
Although such an omniscient attacker is unrealistic in practice, this assumption allows us to analyze the worst-case scenario.

Define the poisoned output sequence as
\begin{align}
\widetilde{\textbf{y}} \triangleq \col\left(
\textbf{y}(0),\ldots,\textbf{y}(N\!-\!1)
\right) \!+\!  \col\left(
\textbf{a}(0),\ldots,\textbf{a}(N\!-\!1)
\right),\!\!
\end{align}
where $ \textbf{a}(k) \triangleq [~\textbf{a}_1(k),\ldots,\textbf{a}_p(k)~]^\top \in \R^p $ represents a poisoning attack against the output data at time $ k $.
Since the data in protected storage are secure, $ \textbf{a}_i \equiv 0 $ and $ \widetilde{\textbf{y}}_i = \textbf{y}_i $ for all $ i \in \cS $.
Denote the index set of the compromised storage as $ \A  \triangleq \left\lbrace i \in \U:\textbf{a}_i \not\equiv 0\right\rbrace $.
If all unprotected storage locations are compromised, then $ \A = \U $.

We assume that the system parameters $ \bar A $, $ \bar B $, and $ \bar C $ are unknown, whereas the I/O data $ (\textbf{u}, \widetilde{\textbf{y}}) $, and the protected/unprotected storage information are known.
The attack sequence $ \textbf{a} $, the attack-free output data $ \textbf{y}_\U $ in unprotected storage $ \U $, and the compromised storage set $ \A $ are all unknown.
In this paper, we focus on offline poisoning attacks against output data, whereas the online measurements obtained during controller operation are assumed to be attack-free\footnote{Several studies propose data-driven controllers against malicious \textit{online} attacks.
See, e.g., \cite{2023TACBullo,2025AutomaticaLiu,2025L-CSSHespanha}.
Note that these studies considered only online attacks and do not account for offline poisoning attacks studied here.}. 
Secure control problems in the presence of both offline/online attacks are important topics for future work.



\subsection{Data Representation}

For a given set $ \Gamma \subseteq \cP $, denote
\begin{align}
	\label{eq:y_tilde_Gamma}
\widetilde{\textbf{y}} _{\Gamma}\triangleq \col(\widetilde{\textbf{y}}_\Gamma(0),\ldots,\widetilde{\textbf{y}}_\Gamma (N-1)) \in \R^{|\Gamma|N},
\end{align}
where $ \widetilde{\textbf{y}}_\Gamma(k) \in \R^{|\Gamma|} $ is the subvector obtained from $ \widetilde{\textbf{y}}(k) $ indexed by the set $ \Gamma $.
Hence, $ \widetilde{\textbf{y}}_{\cS} $ and $ \widetilde{\textbf{y}}_{\U} $ denote the output data stored in protected storage and unprotected storage, respectively.
For notational simplicity, denote the full I/O data by $ \textbf{D} \triangleq(\textbf{u}, \widetilde{\textbf{y}} ) = (\textbf{u}, \widetilde{\textbf{y}}_\cS,\widetilde{\textbf{y}}_\mathcal{U}  ) $.
Also, define the I/O data in the protected storage (i.e., protected data) as $ \textbf{D}_\cS \triangleq (\textbf{u}, \widetilde{\textbf{y}}_\cS ) = (\textbf{u}, \textbf{y}_\cS ) \subseteq \textbf{D}$.
Note that $ \textbf{D} = (\textbf{D}_\cS, \widetilde{\textbf{y}}_\mathcal{U})$.

Given two nonnegative integers $ T_\p $ and $ T_\f $, we partition the full I/O data into two parts, referred to as the \textit{past} and \textit{future data}, as
\begin{align}
\label{eq:Hankel-u}
\left[\!\!
\begin{array}{c}
\textbf{U}^\p \\ \textbf{U}^\f
\end{array}\!\!\right] &\triangleq \mathscr{H}_{T_\p+ T_\f}\left(\textbf{u}\right) \in \R^{m(T_\p + T_\f)\times(N-(T_\p + T_\f)+1)},~\\
\label{eq:Hankel-y}
\left[\!\!
\begin{array}{c}
\widetilde{\textbf{Y}}^\p \\ \widetilde{\textbf{Y}}^\f
\end{array}\!\!\right] &\triangleq \mathscr{H}_{T_\p + T_\f}\left(\widetilde{\textbf{y}}\right)\in \R^{p(T_\p + T_\f)\times(N-(T_\p + T_\f)+1)},
\end{align}
where $ \textbf{U}^\p $ consists of the first $ T_\p $ block rows of $ \mathscr{H}_{T_\p+ T_\f}\left(\textbf{u}\right) $ and $ \textbf{U}^\f $ consists of the last $ T_\f $ block rows of the Hankel matrix (similarly for $ \widetilde{\textbf{Y}}^\p $ and $ \widetilde{\textbf{Y}}^\f $).
For a set $ \Gamma \subseteq \cP $, define the Hankel matrix associated with $ \widetilde{\textbf{y}} _{\Gamma} $ as
\begin{align}
\label{eq:Hankel-y_Gamma}
\left[\!\!
\begin{array}{c}
\widetilde{\textbf{Y}}^\p_\Gamma \\ \widetilde{\textbf{Y}}^\f_\Gamma
\end{array}\!\!\right] \triangleq \mathscr{H}_{T_\p + T_\f}\left(\widetilde{\textbf{y}}_\Gamma\right)\in \R^{|\Gamma|(T_\p + T_\f)\times(N-(T_\p + T_\f)+1)}.
\end{align}

\subsection{An Overview of DeePC \cite{2019ECCDorfler}}
Consider the attack-free scenario (i.e., $ \textbf{a}\equiv 0 $).
Then, the DeePC problem in the online control phase for (\ref{eq:system_model}) at time $ k \in \Z^+ $ can be formulated as follows:
\begin{subequations}
\label{eq:DeePC}
\begin{align}
& \minimize_{g,u^\f,y^\f}~~J\left(u^\f,y^\f\right) \\
\label{eq:DeePC_const}
&\mathrm{s.t.}~\left[\!\!
\begin{array}{c}
\textbf{U}^\p \\ \textbf{Y}^\p \\ \textbf{U}^\f \\ \textbf{Y}^\f
\end{array}\!\!\right]g = \left[\!\!
\begin{array}{c}
u^\mathsf{p} \\ y^\mathsf{p} \\ u^\f \\ y^\f
\end{array}\!\!\right],~u^\f \in \cC_u^{T_\f},~y^\f \in \cC_y^{T_\f},
\end{align}
\end{subequations}
\noindent
where $ \cC_u^{T_\f} $ is the $ T_\f $-fold Cartesian product of $\cC_u $ (similarly for $ \cC_y^{T_\f} $), and $ u^\p = u^{[k-T_\p, k-1]}$ and $ y^\p = y^{[k-T_\p, k-1]} $ are the most recent $ T_\p $-step past I/O online measurements of the system (\ref{eq:system_model}), respectively.
At the activation time $ k = 0 $, it is assumed that the controller has access to the past online measurements $ u^{[-T_\p, -1]} $ and $ y^{[-T_\p,-1]} $.

The decision variables are defined as $ g \in \R^{N- (T_\mathsf{p} + T_\f)  +1} $, $ u^\f \triangleq \col(u^\f(k),\ldots,u^\f(k+T_\f-1)) \in \R^{mT_\f} $, and $ y^\f \triangleq \col(y^\f(k),\ldots,y^\f(k+T_\f-1)) \in \R^{pT_\f}$.
The objective function is defined as
\begin{align*}
J\left(u^\f,y^\f\right)\! \!\triangleq\! \sum_{t=k}^{k+T_\f-1}\left(\left\| y^\f(t)\!-\!r^y(t)\right\|_Q^2\!+\!\left\| u^\f(t)\!-\!r^u(t)\right\|_R^2\right),
\end{align*}
where $ r^y(k) \in\R^p $ and $ r^u(k) \in\R^m  $ denote the desired output and input references at time $ k $, respectively, and $ Q \in \R^{p \times p} $ and $ R \in \R^{m\times m} $ are positive definite matrices.

For each $ k \in \Z^+ $, the DeePC problem (\ref{eq:DeePC}) is computed, and the given input sequence for the unknown system is applied as $ u(k),\ldots,u(k+s) = u^\f(k),\ldots,u^\f(k+s)$ for some $ s < T_\f -1 $.
Then, setting $ k $ to $ k+s $ and updating $ u^\p $ and $ y^\p $, we can iterate the DeePC algorithm for all $ k $.

If $ (\bar A, \bar B) $ is controllable, $ (\bar A, \bar C) $ is observable, $ T_\p \geq n$, and $ \textbf{u} $ is PE of order $ n + T_\p + T_\f $, then the optimal control sequence and corresponding system output obtained by DeePC coincide with those of the MPC problem with prediction horizon $ T_\f $ without identifying any model parameters \cite[Corollary 5.1]{2019ECCDorfler}.
However, when the data are poisoned, the Hankel matrices $ \textbf{Y}^\p $ and $ \textbf{Y}^\f $ are corrupted by the attack, which may prevent the computation of optimal inputs and may even destabilize the system.
This motivates the development of Secure DeePC.

\section{Fundamental Limitations of Data-Driven Control using Poisoned Data}
\label{section:limitation}
Before deriving Secure DeePC, we show the fundamental limitations of data-driven control using poisoned data, where an adversary can arbitrarily corrupt the unprotected data.

\subsection{Impossibility of Certified Detection and Identification}
\label{subsection:detection_identification}
Let $ \mathscr{D}:\textbf{D} \mapsto\{\mathrm{Attack,~NoAttack}\} $ and $ \mathscr{I}:\textbf{D} \mapsto 2^\cP $ denote an arbitrary data-driven attack detector and identifier, respectively, where $ 2^\cP $ denotes the power set of $ \cP $.
The following result provides the impossibility of certified data-driven attack detection and identification.
\begin{proposition}
\label{proposition:impossibility_detection}
Suppose that there exists at least one unprotected storage location, i.e., $ \U \neq \emptyset $.
There exist an attack-free dataset $ \textbf{D}_0 $, a minimal system $ (A',B',C') $, and a nonzero poisoning attack sequence $ \textbf{a} $ yielding a poisoning dataset $ \textbf{D} $ such that $ \textbf{D}_0 = \textbf{D} $.
Consequently, for \textit{any} data-driven attack detector $ \mathscr{D} $ and identifier $ \mathscr{I}$, $ \mathscr{D}(\textbf{D}_0) = \mathscr{D}(\textbf{D}) $ and $ \mathscr{I}(\textbf{D}_0) = \mathscr{I}(\textbf{D})$, which implies the attack $ \textbf{a} $ cannot be detected and identified.
\end{proposition}
\begin{proof}
	Let $ \textbf{D}_0 \triangleq (\textbf{u}, \textbf{y})$ be the attack-free dataset such that $ \textbf{y}_j \not\equiv 0 $ for some $ j \in \cP $.
	Let $ \textbf{x}(\cdot) $ be the offline state trajectory of $ (\bar A,\bar B) $ under input $ \textbf{u} $ and initial condition $ \textbf{x}(0) $.
	Fix any $ i \in \U $ and, for $ \alpha > 0 $, define $ T_\alpha \triangleq I_p + \alpha e_i e_j^\top $ and $ C' \triangleq T_\alpha \bar C $, where $ e_i $ and $ e_j $ are the $ i $th and $ j $th canonical basis of $ \R^p $, respectively.
	
	Consider another system $ (A', B', C') $ with $ A' \triangleq \bar A $ and $ B' \triangleq \bar B $.
	Since $ T_\alpha $ is invertible, $ (A',C') $ is observable, and thus $ (A', B', C') $ is minimal.
	Consider a nonzero attack sequence $ \textbf{a} $ such that 
	\begin{align}
		\label{eq:attack_proposition}
		\textbf{a}_\ell(k)=\left\lbrace \!\!
		\begin{array}{ll}
			-\alpha \textbf{y}_j(k),&~\ell = i, \\
			0, &~\mathrm{otherwise}.
		\end{array}\right.
	\end{align}
	Then, the $ i $th poisoned output $ \widetilde{\textbf{y}}'_i $ of $ (A', B', C') $ is given by
	\begin{align*}
		\widetilde{\textbf{y}}'_i(k)&=C'_i \textbf{x}(k) + \textbf{a}_i(k)\\
		&= \left(\bar C_i + \alpha\bar C_j\right)\textbf{x}(k) -\alpha \textbf{y}_j(k) = \bar C_i \textbf{x}(k) = \textbf{y}_i(k).
	\end{align*}
	Also, for $ \ell \in \cP \setminus \{i\} $, we have $ \widetilde{\textbf{y}}'_\ell(k) = \textbf{y}_\ell(k) $, which implies $ \widetilde{\textbf{y}} \equiv  \textbf{y}$; therefore, $ \textbf{D}_0 = \textbf{D} $.
	Consequently, for any data-driven attack detector $ \mathscr{D} $ and identifier $ \mathscr{I}$, $ \mathscr{D}(\textbf{D}_0) = \mathscr{D}(\textbf{D}) $ and $ \mathscr{I}(\textbf{D}_0) = \mathscr{I}(\textbf{D})$.
\end{proof}

This proposition implies that any detector $ \mathscr{D} $ declares $ \mathrm{NoAttack} $ on some attack-free dataset necessarily declares $ \mathrm{NoAttack} $ on some poisoned dataset as well.
Similarly, for any identifier $ \mathscr{I} $, if it returns $ \emptyset $ on some attack-free dataset, then it also returns $ \emptyset $ on some poisoned dataset.
The result establishes the fundamental impossibility of certified detection and identification of poisoning attacks based solely on poisoned data.
Hence, even if some data stored in an unprotected location are attack-free (i.e., $ |\A|<|\U| $), it is impossible to identify them with certainty, which implies that creating an additional attack-free dataset from unprotected storage is impossible.
In order to detect and identify poisoning attacks using only data, additional assumptions on the attacks and/or adversaries are required, such as the existence of a sufficiently long attack-free interval in the dataset \cite{2026ShinoharaTCNS}.

Importantly, this impossibility is not alleviated by assuming a small number of attacked channels. 
Also, knowing the maximal number of corrupted channels $ |\A| $ does not remove this impossibility. 
This stands in sharp contrast to model-based settings, where the measurement structure  (i.e., the matrix $ \bar C $) is known and redundancy conditions, such as sparse observability, enable attack detection and identification \cite{2016TACTabuada}.
On the other hand, in the data-driven setting, the underlying measurement structure is unknown; thus, redundancy in the number of output channels does not yield certifiable guarantees under arbitrary poisoning attacks.


\subsection{Non-informativity of Unprotected Data}
We next show that the unprotected data $ \widetilde{\textbf{y}}_\U $ are non-informative for data-driven control.
To this end, we first define the set of systems consistent with the full data $ \textbf{D} $ as
\begin{align}
\label{eq:Sigma}
&\Sigma\left(\textbf{D}\right) \triangleq \left\lbrace (A,B,C):\exists \textbf{x}_0, \exists \textbf{a} \in\R^{pN}~\mathrm{s.t.}~\textbf{x}(0)=\textbf{x}_0, \right. \nonumber \\
&~~~~\textbf{x}(k+1)=A\textbf{x}(k)+ B \textbf{u}(k), ~\widetilde{\textbf{y}}_\cS(k) = C_{\cS}\textbf{x}(k),\nonumber \\
& ~~~~\left. \widetilde{\textbf{y}}_{\U}(k) = C_{\U}\textbf{x}(k) + \textbf{a}_{\U}(k),~\forall k \in [0,N-1]\right\rbrace,
\end{align}
where $ C_{\cS} \in \R^{|\cS|\times n}$ denotes the submatrix obtained from $ C $ indexed by $ \cS $ (similar for $ C_{\U} $).
Similarly, the set of systems consistent with the protected data $ \textbf{D}_\cS $ is given as
\begin{align}
\label{eq:Sigma_S}
&\Sigma_\cS\left(\textbf{D}_\cS\right) \triangleq \left\lbrace (A,B,C):\exists \textbf{x}_0 ~\mathrm{s.t.}~\textbf{x}(0)=\textbf{x}_0, \right. \nonumber \\
& ~~~~\textbf{x}(k+1)=A\textbf{x}(k)+ B \textbf{u}(k),~\widetilde{\textbf{y}}_\cS(k)=C_\cS\textbf{x}(k), \nonumber \\
& ~~~~\left.\forall k \in [0,N-1]\right\rbrace.
\end{align}
Then, we have the following result.
\begin{proposition}
\label{proposition:model_equivalence}
For given $ \textbf{D} $, $ \Sigma_\cS(\textbf{D}_\cS) = \Sigma(\textbf{D}) $.
\end{proposition}
\begin{proof}
	The inclusion $ \Sigma(\textbf{D})\subseteq \Sigma_\cS(\textbf{D}_\cS) $ is immediate from the definitions.
	Conversely, let $ (A,B,C) \in \Sigma_\cS(\textbf{D}_\cS) $, and let $ \textbf{x}(\cdot) $ be a trajectory certifying this membership.
	Define $ \textbf{a}_\U(k) \triangleq \widetilde{\textbf{y}}_{\U}(k) - C_\U \textbf{x}(k) $ with $ \textbf{a}_{\cS}(k) \equiv 0 $.
	Then, $ \widetilde{\textbf{y}}_\U(k) = C_{\U}\textbf{x}(k) + \textbf{a}_{\U}(k) $, so $ (A,B,C) \in \Sigma(\textbf{D}) $.
	Hence, $ \Sigma_\cS(\textbf{D}_\cS) \subseteq \Sigma(\textbf{D}) $, and therefore $ \Sigma_\cS(\textbf{D}_\cS) = \Sigma(\textbf{D}) $.
%
\end{proof}

Under arbitrary poisoning attacks, the set of systems consistent with the full data $ \textbf{D} $ matches the set of systems expressible using only the protected data $ \textbf{D}_\cS $, irrespective of the particular realization of the unprotected data $ \widetilde{\textbf{y}}_\U $.
Building on this property, we next show the non-informativity of the unprotected data.

Let $ \Pi $ denote the class of admissible controllers in the system (\ref{eq:system_model}), and let $ \mathscr{A}: \textbf{D} \mapsto \pi \in \Pi $ denote an arbitrary data-driven design rule that outputs a controller $ \pi $.
For all protected datasets $ \textbf{D}_\cS $, define a data-driven design rule $ \mathscr{A}_\cS: \textbf{D}_\cS \mapsto \pi \in \Pi $ based on protected data as $ \mathscr{A}_\cS(\textbf{D}_\cS) \triangleq \mathscr{A}((\textbf{D}_\cS, 0))$.
Denote an arbitrary loss function for a data-driven controller by $ \mathscr{L}: \Pi \times \Sigma(\textbf{D}) \rightarrow \R \cup \{+\infty\} $.
Then, we have the following result.

\begin{proposition}
\label{proposition:algorithm_equivalence}
For given protected data $ \textbf{D}_\cS $,
\begin{align}
\label{eq:algorithm_equivalence}
&\inf_{\mathscr{A}} \sup_{\widetilde{\textbf{y}}_\U} \sup_{(A,B,C)\in \Sigma((\textbf{D}_\cS, \widetilde{\textbf{y}}_\U ))} \mathscr{L}\left(\mathscr{A}\left((\textbf{D}_\cS, \widetilde{\textbf{y}}_\U)\right), \left(A,B,C\right)\right) \nonumber\\
&~~~=\inf_{\mathscr{A}_\cS} \sup_{(A,B,C)\in \Sigma_\cS(\textbf{D}_\cS)} \mathscr{L}\left(\mathscr{A}_\cS\left(\textbf{D}_\cS\right), \left(A,B,C\right)\right).
\end{align}
\end{proposition}
\begin{proof}
	For notational clarity, write $ \theta = (A,B,C) $ for a candidate system. 
	By Proposition~\ref{proposition:model_equivalence}, $ \Sigma_\cS(\textbf{D}_\cS) = \Sigma(\textbf{D}) $ for all $ \widetilde{\textbf{y}}_\U $, which implies that the left-hand side of (\ref{eq:algorithm_equivalence}) is
	\begin{align*}
		\inf_{\mathscr{A}} \sup_{\widetilde{\textbf{y}}_\U} \sup_{\theta \in \Sigma_\cS(\textbf{D}_\cS)} \mathscr{L}\left(\mathscr{A}\left( (\textbf{D}_\cS, \widetilde{\textbf{y}}_\U)\right), \theta \right).
	\end{align*}
	Fix any design rule $ \mathscr{A}$.
	Since the supremum over $ \widetilde{\textbf{y}}_\U $ dominates the value at $ \widetilde{\textbf{y}}_\U = 0 $, it follows that
	\begin{align*}
		\sup_{\widetilde{\textbf{y}}_\U} \!\!\sup_{\theta \in \Sigma_\cS(\textbf{D}_\cS)} \!\!\!\!\!\mathscr{L}\left(\mathscr{A}\left((\textbf{D}_\cS, \!\widetilde{\textbf{y}}_\U)\right), \!\theta \right) \!\geq\!\!\!\!  \sup_{\theta\in \Sigma_\cS(\textbf{D}_\cS)} \!\!\!\!\!\mathscr{L}\left(\mathscr{A}\left((\textbf{D}_\cS , \!0)\right), \!\theta \right).
	\end{align*}
	Hence, taking the infimum over $ \mathscr{A} $ on both sides and using $ \mathscr{A}((\textbf{D}_\cS,0 )) = \mathscr{A}_\cS(\textbf{D}_\cS) $, this can be written as
	\begin{align}
		\label{eq:proposition_inequality_02}
		&\inf_{\mathscr{A} }\sup_{\widetilde{\textbf{y}}_\U} \sup_{\theta\in \Sigma_\cS(\textbf{D}_\cS)} \mathscr{L}\left(\mathscr{A}\left((\textbf{D}_\cS, \widetilde{\textbf{y}}_\U)\right), \theta \right) \nonumber \\
		& \hspace{10mm} \geq \inf_{\mathscr{A}_\cS } \sup_{\theta\in \Sigma_\cS(\textbf{D}_\cS)} \mathscr{L}\left(\mathscr{A}_\cS\left(\textbf{D}_\cS\right), \theta \right).
	\end{align}
	
	Conversely, fix any design rule $ \mathscr{A}_\cS $, and define its full-data extension $ \mathscr{A}^e $ by $ \mathscr{A}^e(\textbf{D}'_\cS, \widetilde{\textbf{y}}_{\cS}) \triangleq \mathscr{A}_\cS(\textbf{D}'_\cS)$ for all $ \textbf{D}'_\cS $ and $ \widetilde{\textbf{y}}_\U $, namely, $ \mathscr{A}^e $ ignores the unprotected data.
	Then, 
	\begin{align*}
		\sup_{\widetilde{\textbf{y}}_\U} \!\! \sup_{\theta \in \Sigma_\cS(\textbf{D}_\cS)}\!\!\!\! \mathscr{L}\left(\mathscr{A}^e\left((\textbf{D}_\cS, \widetilde{\textbf{y}}_\U)\right), \theta \right) =\!\!\!\! \sup_{\theta\in \Sigma_\cS(\textbf{D}_\cS)} \!\!\!\!\mathscr{L}\left(\mathscr{A}_\cS\left(\textbf{D}_\cS\right), \theta \right)\!.
	\end{align*}
	Hence, we have
	\begin{align*}
		&\inf_{\mathscr{A}} \sup_{\widetilde{\textbf{y}}_\U} \sup_{\theta \in \Sigma_\cS(\textbf{D}_\cS)} \mathscr{L}\left(\mathscr{A}\left( (\textbf{D}_\cS, \widetilde{\textbf{y}}_\U)\right), \theta \right)\\
		&\hspace{10mm} \leq\sup_{\widetilde{\textbf{y}}_\U}  \sup_{\theta \in \Sigma_\cS(\textbf{D}_\cS)} \mathscr{L}\left(\mathscr{A}^e\left((\textbf{D}_\cS, \widetilde{\textbf{y}}_\U)\right), \theta \right) \\
		&\hspace{10mm} = \sup_{\theta\in \Sigma_\cS(\textbf{D}_\cS)} \mathscr{L}\left(\mathscr{A}_\cS\left(\textbf{D}_\cS\right), \theta \right)
	\end{align*}
	Now that $ \mathscr{A}_\cS $ is arbitrary, taking the infimum over $ \mathscr{A}_\cS  $ yields the reverse relation of (\ref{eq:proposition_inequality_02}).
	Therefore, (\ref{eq:algorithm_equivalence}) holds. 
\end{proof}

In a minimax formulation where the system operator seeks to minimize the worst-case performance over all systems consistent with the stored data and over all admissible poisoning attacks, Proposition~\ref{proposition:algorithm_equivalence} indicates that exploiting $ \widetilde{\textbf{y}}_\U $ cannot improve the worst-case performance beyond what is achievable using $ \textbf{D}_\cS $ alone.
Therefore, for any data-driven controller with worst-case guarantees, designing it solely using the protected data is optimal.

\subsection{Impossibility of Certified Guarantee for Output Safety}
We next show that hard safety constraints on outputs $ \mathcal{C}_y $ cannot be guaranteed from the poisoned dataset. 
\begin{proposition}
	\label{proposition:constraints}
	Suppose $ \U \neq \emptyset $, and the output constraint set $ \mathcal{C}_y $ imposes a bounded interval on some component $ i \in \U $.
	Let $ \textbf{D} $ be any offline poisoned dataset generated by the true plant $ (\bar A, \bar B, \bar C) $ with $ (\bar A, \bar C_\cS) $ observable and any attack sequence $ \textbf{a} $.
	Assume that the given $ T_\p $-step online sequence $ (u^{[-T_\p, -1]}, y^{[-T_\p, -1]}) $ has terminal state $ x^\star $ satisfying $ \bar C_j x^\star \neq 0 $ for some $ j \in \cS $.
	Then, there exists a minimal system $ (A', B', C') \in \Sigma(\textbf{D}) $ which leads to the $ i $th output constraint's violation in online control operation.
	Consequently, no data-driven design rule based solely on $ \textbf{D} $ can guarantee the hard constraint on outputs in $ \U $ uniformly over the consistency class $ \Sigma(\textbf{D}) $.
\end{proposition}
\begin{proof}
	Let $ [\overline{c}_i, \underline{c}_i] $ be the bounded interval constraint imposed on the $ i $th output.
	Let $ (u^{[-T_\p, -1]}, y^{[-T_\p, -1]}) $ be the assumed online sequence of the true plant, and let $ x^\star = x(0) $ be its terminal state satisfying $ \bar C_j x^\star \neq 0 $ for some $ j \in \cS $.
	As with Proposition~\ref{proposition:impossibility_detection}, for $ \alpha \neq 0 $, define $ T_\alpha \triangleq I_p + \alpha e_i e_j^\top $ and $ C' \triangleq T_\alpha \bar C $.
	Since $ i \in \U $ and $ j \in \cS $, it holds that $ C'_\cS = \bar C_\cS $.
	Considering $ A' \triangleq \bar A $ and $ B' \triangleq \bar B $, $ (A', B', C') $ is minimal.
	
	Let $ \textbf{x}(\cdot) $ be the offline state trajectory of $ (\bar A,\bar B) $ that generates the protected dataset $ \textbf{D}_\cS $ under input $ \textbf{u} $ and initial condition $ \textbf{x}(0) $.
	Since the protected rows are unchanged, the same trajectory also certifies that $ (A',B',C') \in \Sigma_\cS(\textbf{D}_\cS) $.
	By Proposition~\ref{proposition:model_equivalence}, $ \Sigma_\cS(\textbf{D}_\cS) = \Sigma(\textbf{D}) $ yields $ (A', B', C')\in \Sigma(\textbf{D}) $.
	
	Next, consider the same initial state at time $ -T_\p $ as in the true plant and apply the same online input sequence $ u^{[-T_\p, -1]} $.
	Since the state dynamics are unchanged, the resulting online state trajectory is identical to that of the true plant, and in particular, it reaches the same terminal state $ x^\star $ at time $ 0 $.
	Moreover, because $ C'_\cS = \bar C_\cS $, the protected online outputs over $ [-T_\p, -1] $ are unchanged.
	Hence, the same protected online sequence $ (u^{[-T_\p, -1]}, y^{[-T_\p, -1]}_\cS) $ is also consistent with $ (A',B',C') $.
	
	Under $ (A',B',C') $, the $ i $th output at time $ 0 $ is $ y'_i(0) = C'_i x^\star = \left(\bar C_i + \alpha \bar C_j \right)x^\star $.
	Since $ \bar C_j x^\star \neq 0 $, we can choose $ \alpha $ with sufficiently large magnitude and appropriate sign so that $ y'_i(0) \notin [\overline{c}_i, \underline{c}_i] $, which implies the violation of the $ i $th output constraint at time $ 0 $.
\end{proof}

This proposition does not assert that the true plant necessarily violates the hard constraint. 
Rather, it shows that the dataset $ \textbf{D} $ cannot exclude consistent plants that do. 
This is the precise sense in which hard safety on unprotected outputs is not guaranteed under arbitrary offline poisoning attacks.
This suggests a security-by-design principle: every output channel that enters hard safety constraints should be included in the protected storage.

\section{Secure DeePC}
\label{section:SecureDeePC}

Proposition \ref{proposition:algorithm_equivalence} justifies that, under arbitrary offline poisoning attacks, ignoring the unprotected data $ \widetilde{\textbf{y}}_\U $ is a minimax-optimal choice.
We turn this limitation into a design principle by proposing \textit{Secure DeePC}, whose basic ideas can be summarized as follows:
\begin{itemize}
	\item Until the online input data become PE, compute the output-truncated DeePC that uses only the protected dataset $ \textbf{D}_\cS $, taking into account the fundamental limitations of Propositions~\ref{proposition:algorithm_equivalence}--\ref{proposition:constraints}. In this phase, a small dither signal is added to the input to enhance PE.
	\item Once the online data become PE, construct a linear mapping from $ (u,y_\cS) $ to $ y_\U $. Using this linear mapping, reconstruct the offline, attack-free data within a specified interval. Simultaneously, detect and identify any attacks present in the offline data during the interval.
	\item After the online input data become PE, execute the original DeePC (\ref{eq:DeePC}) using the reconstructed offline data. 
\end{itemize}

\begin{algorithm}[t]
	\caption{Secure DeePC}
	\label{algorithm:Secure_DeePC}
	\begin{algorithmic}[1]
		\Statex \hspace{-6mm} \textbf{Input:} offline data $ \textbf{D} = (\textbf{u}, \widetilde{\textbf{y}}_\cS,\widetilde{\textbf{y}}_\mathcal{U}  )  $, $ \bar \cC_u $, $ \bar \cC_{y} $, $ \cC_u $, $ \cC_y $, $T_\p$, $T_\f $, objective functions $ J_\cS $ and $ J $, and $ \W \subset\R^m $.
		\State $ \tau_{\rm pe} \leftarrow +\infty $.
		\State Construct $ (\textbf{U}^\p, \widetilde{\textbf{Y}}^\p_\cS, \textbf{U}^\f, \widetilde{\textbf{Y}}^\f_\cS) $ using $ (\textbf{u}, \widetilde{\textbf{y}}_\cS) $.
		\For{each $ k \in \Z^+ $}
		\State Retrieve past I/O data $ (u^\p, y^\p) $.
		\If{$ k < \tau_{\rm pe} $}
		\If{$ u^{[-T_\p,k-1]} $ is PE of order $ 2T_\p $}
		\State Set $ \tau_{\rm pe} \leftarrow k $; $ \tau_{\rm sw} \leftarrow k+1 $.
		\EndIf
		\State Compute (\ref{eq:Secure_temp_DeePC}) using $ (\textbf{U}^\p, \widetilde{\textbf{Y}}^\p_\cS, \textbf{U}^\f, \widetilde{\textbf{Y}}^\f_\cS) $ for $ u^\f $. 
		\State Apply $ u(k) \!=\! u^\f(k) + w(k) $, where $ w(k) $ is an i.i.d. random signal satisfying $ w(k) \!\in\! \mathcal{W} $.
		\EndIf
		\If{$ k = \tau_{\rm sw} $}
		\State Compute $ F(\tau_{\rm sw} ) $ as $ F(\tau_{\rm sw} ) = Y_\U(\tau_{\rm sw} ) Z(\tau_{\rm sw} )^\dagger $. 
		\State Recover $ \widehat{\textbf{y}}_\U(k)=F(\tau_{\rm sw} )\textbf{z}(k) $ for $ k \in [T_\p,N\!-\!1] $. 
		\State Construct $ (\widehat{\textbf{U}}^\p, \widehat{\textbf{Y}}^\p,\widehat{\textbf{U}}^\f, \widehat{\textbf{Y}}^\f) $ using the \textit{recovered} dataset $ \widehat{\textbf{D}} = (\textbf{u}^{[T_\p,N-1]},\widetilde{\textbf{y}}^{[T_\p, N-1]}_\cS, \widehat{\textbf{y}}^{[T_\p, N-1]}_\U) $.
		\State Identify the attacks as $ \widehat{\textbf{a}}_\U(k) = \widetilde{\textbf{y}}_\U(k) - \widehat{\textbf{y}}_\U(k) $. 
		\EndIf
		\If{$ k \geq \tau_{\rm sw} $}
		\State Compute (\ref{eq:DeePC}) using $ (\widehat{\textbf{U}}^\p, \widehat{\textbf{Y}}^\p,\widehat{\textbf{U}}^\f, \widehat{\textbf{Y}}^\f) $ for $ u^\f $. 
		\State Apply $ u(k) = u^\f(k) $.
		\EndIf
		\EndFor
	\end{algorithmic}
\end{algorithm}	

Algorithm~\ref{algorithm:Secure_DeePC} presents the Secure DeePC framework.
As with the original DeePC, at the activation time $ k=0 $, the past measurements $ u^{[-T_\p, -1]} $ and $ y^{[-T_\p, -1]} $ are given.
This algorithm has a two-phase structure.
In Phase I $ (k < \tau_{\rm sw} ) $, the controller relies exclusively on the protected dataset $ \textbf{D}_{\cS} $ and enforces only the protected-output constraints, following Propositions~\ref{proposition:algorithm_equivalence}--\ref{proposition:constraints}.
The time $ \tau_{\rm sw} $, which is defined as $ \tau_{\rm sw} \triangleq \tau_{\rm pe} + 1 $, indicates the switching time from Phase I to Phase~II, where $ \tau_{\rm pe} $ denotes the first time at which the online input data $ u^{[-T_\p, \tau_{\rm pe}-1]} $ is PE of order $ 2T_\p $.
Then, in Phase II $ (k \geq \tau_{\rm sw}) $, Secure DeePC learns the linear mapping $ \bar F $ from $ (u,y_\cS) $ to $ y_\U $ and uses it to reconstruct the attack-free offline unprotected outputs; then it solves the full-output DeePC (\ref{eq:DeePC}) using the recovered dataset.

In what follows, we explain the details of this algorithm and provide some theoretical results.

\subsection{$ k < \tau_{\rm sw} :$ Truncated DeePC with Active Input Excitation}
\label{subsection:first_phase}

Our starting point is the output-truncated DeePC that uses Hankel matrices built from the protected data $ \textbf{D}_\cS $:
\begin{subequations}
	\label{eq:Secure_temp_DeePC}
	\begin{align}
		& \minimize_{g,u^\f,y^\f_\cS}~~J_\cS\left(u^\f,y^\f_\cS\right) \\
		\label{eq:Secure_temp_DeePC_const}
		&\mathrm{s.t.}~\left[\!\!
		\begin{array}{c}
			\textbf{U}^\p \\ \widetilde{\textbf{Y}}^\p_\cS \\ \textbf{U}^\f \\ \widetilde{\textbf{Y}}^\f_\cS
		\end{array}\!\!\right]g = \left[\!\!
		\begin{array}{c}
			u^\mathsf{p} \\ y^\mathsf{p}_\cS \\ u^\f \\ y^\f_\cS
		\end{array}\!\!\right],~u^\f \in \bar \cC_u^{T_\f},~y^\f_\cS \in \bar \cC_{y}^{T_\f},
	\end{align}
\end{subequations}
where $ y^\p_\cS = y_\cS^{[k-T_\p,k-1]}$ is the $ T_\p $ most recent past output measurements in the set $ \cS $ from the system (\ref{eq:system_model}).
Problem~(\ref{eq:Secure_temp_DeePC}) imposes constraints only on input and protected output components, following Proposition~\ref{proposition:constraints}.
We define the induced constraint on the protected outputs as
\begin{align*}
	\bar \cC_{y} \triangleq \mathrm{proj}_{\cS}\left(\cC_{y}\right) = \left\lbrace y_\cS \in \R^{|\cS|}: \exists \psi \in \cC_y~\mathrm{s.t.}~y_\cS = \psi_\cS \right\rbrace.
\end{align*}
The set $ \bar{\cC}_{y}^{T_\f} $ is the $ T_\f $-fold Cartesian product of $\bar \cC_{y} $.
The decision variable $ y^\f_\cS $ is given as $ y^\f_\cS \triangleq \col(y^\f_\cS(k),\ldots,y_\cS^\f(k+T_\f-1)) \in \R^{|\cS|T_\f}$.
The objective function $ J_\cS(u^\f,y^\f_\cS) $ is also defined by only considering the outputs in $ \cS $.

This output-truncated DeePC implementation is intuitive, motivated by the fundamental limitations; however, because constraints and references concerning unprotected components are not applied, performance and safety cannot be ensured over time.
Hence, in our Secure DeePC framework, we reconstruct (sanitize) the offline dataset and then run the original DeePC on the reconstructed data to ensure the same performance and safety with MPC within a finite time.

As shown next, once the online input data is PE of order $ 2T_\p $, one can recover the offline, attack-free, unprotected data $ \textbf{y}_\U $ under certain conditions.
To achieve the excitation condition as soon as possible, we adopt the active input excitation strategy.
Specifically, the input in Phase I is computed as $ u(k) = u^\f(k) + w(k) $,
where $ u^\f(k) $ is computed by (\ref{eq:Secure_temp_DeePC}) with the input constraint $ \bar \cC_{u} \triangleq \cC_{u} \ominus \W $, where $ \ominus $ denotes the Pontryagin set difference, and $ w(k) \in \W $ denotes an i.i.d. random signal with an absolutely continuous distribution supported on $ \W \subset \R^m $ satisfying $ 0 \in \mathrm{int}\,\W $ (e.g., uniform on $ \W $).
This constraint reformulation is needed to ensure input safety even under the active input excitation.
By adding the i.i.d. signal $ w(k) $ to the input, the input sequence $ u^{[-T_\p,k-1]} $ is PE of order $ 2T_\p $ almost surely if $ k \geq 2(m+1)T_\p - 1 $ \cite[Lemma 1]{2026ShinoharaTCNS}.
In practice, the magnitude of $ w(k) $ can be small enough to preserve nominal closed-loop performance, while still providing the excitation required for the rank computation.

\subsection{$ k = \tau_{\rm sw}: $ Offline Data Reconstruction}
\label{subsection:second_phase}
We next present the procedure for reconstructing the offline attack-free dataset.
For each $k\in \Z^+$ in an \textit{online} operation, define the vector based on the most recent $ T_\p $-step online input and output measurements in $ \cS $:
\begin{align*}
	\zeta(k)\triangleq \col\left(u^{[k-T_\p,k-1]},y^{[k-T_\p,k-1]}_\cS \right)\in\mathbb{R}^{(m+|S|)T_\p}.
\end{align*}
Also, define the matrix whose columns are $z(\cdot)$:
\begin{align*}
	Z(k) \!\triangleq\! \left[\!\!\begin{array}{ccc}
		\zeta(0)\! & \!\cdots\! &\!\zeta(k-1)
	\end{array}\!\!\right] \!=\! \left[\!\!\begin{array}{c}
	\mathscr{H}_{T_\p}\left(u^{[-T_\p,k-2]}\right) \\\mathscr{H}_{T_\p}\left(y_\cS^{[-T_\p,k-2]}\right)
\end{array}\!\!\right],
\end{align*}
Define the output matrix corresponding to the set $ \U $:
\begin{align*}
	Y_\U(k) \triangleq \left[\begin{array}{ccc}
		y_{\U}(0) & \!\cdots\! & y_{\U}(k-1)
	\end{array}\right] \in \R^{|\U| \times k}.
\end{align*}
We first show the existence of a linear mapping from $ \zeta $ to $ y_\U $.
\begin{lemma}
	\label{lemma:existence_F}
	Suppose that $(\bar A,\bar C_\cS)$ is observable and $T_\p\geq n$.
	Then, there exists a constant matrix $ \bar F  \in \R^{|\U|\times(m+|\cS|)T_\p}$ such that, for any trajectory of the true plant $ (\bar A, \bar B, \bar C) $, $ y_{\U}(k) = \bar F\zeta(k) $ for all $ k \in \Z^+$.
\end{lemma}
\begin{proof}
	Fix $ k \in \Z^+ $ and let $x^\p\triangleq x(k-T_p)$.
	Then, by the system dynamics (\ref{eq:system_model}), we obtain
	\begin{align}
		y^{[k-T_\p,k-1]}_\cS=\mathcal{O}_{T_\p}x^\p+\mathcal{T}_{T_\p}u^{[k-T_\p,k-1]},
		\label{eq:yS_stack_proof}
	\end{align}
	where $ \mathcal{O}_{T_\p}  \in \R^{|\cS|T_\p \times n}$ is the $ T_\p $-step observability matrix of $ (\bar A, \bar C_\cS) $ and $\mathcal{T}_{T_\p}\in\mathbb{R}^{|\cS|T_\p\times mT_\p}$ is the associated lower-block-Toeplitz matrix.
	Since $(\bar A,\bar C_S)$ is observable and $T_\p\ge n$, $\mathcal{O}_{T_\p}$ has full column rank $n$.
	Hence, multiplying \eqref{eq:yS_stack_proof} by $\mathcal{O}_{T_\p}^\dagger$ gives
	\begin{align}
		x^\p=\mathcal{O}_{T_\p}^\dagger y^{[k-T_\p,k-1]}_\cS-\mathcal{O}_{T_\p}^\dagger\mathcal{T}_{T_\p}u^{[k-T_\p,k-1]}.
		\label{eq:xp_from_up_yp}
	\end{align}
	Unrolling the state dynamics from time $k-T_\p$ to $k$ yields
	\begin{align*}
		x(k)\!=\!\bar A^{T_\p}x^\p\!+\!\mathcal{R}_{T_\p}u^{[k-T_\p,k-1]},~
		\mathcal{R}_{T_\p}\triangleq
		\begin{bmatrix}
			\bar A^{T_\p-1}\bar B \!& \!\cdots\! & \!\bar B
		\end{bmatrix}\!. 
	\end{align*}
	Therefore, using $y_\U(k)=\bar C_\U x(k)$ and substituting \eqref{eq:xp_from_up_yp}, 
	\begin{align}
		y_\U(k)&= \bar C_\U\bar A^{T_\p} x^\p +\bar C_\U\mathcal{R}_{T_\p}u^{[k-T_\p,k-1]} \nonumber\\
		& =\bar F \col\left(u^{[k-T_\p,k-1]},y^{[k-T_\p,k-1]}_\cS\right)\!=\!\bar F \zeta(k)
		\label{eq:yU_equals_Fstar_z}
	\end{align}
	for all $ k\in \Z^+ $, where
	\begin{align*}
		\bar F \triangleq
		\begin{bmatrix}
			\bar C_\U\mathcal{R}_{T_\p}-\bar C_\U\bar A^{T_\p}\mathcal{O}_{T_\p}^\dagger\mathcal{T}_{T_\p}
			&
			\bar C_\U\bar A^{T_\p}\mathcal{O}_{T_\p}^\dagger
		\end{bmatrix}.
	\end{align*}
\end{proof}

The following lemma derives conditions under which the linear map can be constructed from the online measurements.
\begin{lemma}
	\label{lemma:learning_F}
	Suppose that $ (\bar A, \bar B) $ is controllable, $ (\bar A, \bar C_\cS) $ is observable, and $ T_\p \geq n $.
	Assume that the online input $ u^{[-T_\p, \tau_{\rm pe}-1]} $ is PE of order $ 2T_\p $ for some $ \tau_{\rm pe} \in \Z^+ $.
	Define $ F(\tau_{\rm sw}) \triangleq Y_\U(\tau_{\rm sw}) Z(\tau_{\rm sw})^\dagger  $ with $ \tau_{\rm sw} \triangleq \tau_{\rm pe} + 1 $.
	Then, for any trajectory of the true plant $ (\bar A, \bar B, \bar C) $, $ y_{\U}(k) = F(\tau_{\rm sw})\zeta(k) $ for all $ k \in \Z^+$.
\end{lemma}
\begin{proof}
	By Lemma~\ref{lemma:existence_F}, there exists a constant matrix $ \bar F $ such that $ y_{\U}(k) = \bar F \zeta(k) $ for all $ k $.
	Applying this relation column-wise to the online data for $ k \in [0,\tau_{\rm sw}-1] $ yields $ Y_\U(\tau_{\rm sw}) = \bar F Z(\tau_{\rm sw})  $.
	Now let $ \zeta(k) $ be the $ T_\p $-step window generated by an arbitrary trajectory of the same plant.
	Since $ (\bar A, \bar B) $ is controllable, $ T_\p \geq n $, and $ u^{[-T_\p, \tau_{\rm pe}-1]} $ is PE of order $ 2 T_\p \geq n + T_\p $, Willems' Fundamental Lemma \cite{2005SCLWillems} implies that $ \zeta(k) \in \im Z(\tau_{\rm pe}+1) = \im Z(\tau_{\rm sw})$.
	Hence,
	\begin{align*}
		&F(\tau_{\rm sw})\zeta(k) \!=\! Y_\U(\tau_{\rm sw}) Z(\tau_{\rm sw})^\dagger\zeta(k) \\
		&~~~= \bar F Z(\tau_{\rm sw})Z(\tau_{\rm sw})^\dagger\zeta(k) \!=\! \bar F\zeta(k) \!= \! y_{\U}(k),\forall k \in \Z^+.
	\end{align*}
\end{proof}

Under the conditions of Lemma~\ref{lemma:learning_F}, the linear mapping from $ \zeta $ to $ y_\U $ can be learned as  $ F(\tau_{\rm sw}) $ using the online measurements.
Note that, by adding the i.i.d. dither signal $ w(k) $ to the input in Phase I, $ u^{[-T_\p,\tau_{\rm pe}-1]} $ is PE of order $ 2T_\p $ in finite time almost surely, and accordingly, $F(\tau_{\rm sw})$ can be constructed in finite time almost surely.

Using $F(\tau_{\rm sw})$, we can reconstruct the \textit{offline}, unprotected output sequence in the interval $ [T_\p, N-1] $.

\begin{corollary}
	\label{corollary:recovered_U}
	As with $ \zeta(k) $, for the \textit{offline} data, define
	\begin{align*}
		\textbf{z}(k)\triangleq \col\left(\textbf{u}^{[k-T_\p,k-1]},~\textbf{y}^{[k-T_\p,k-1]}_\cS \right)\in\mathbb{R}^{(m+|S|)T_\p},
	\end{align*}
	for $ k \in [T_\p, N-1] $.
	If the same conditions of Lemma~\ref{lemma:learning_F} hold, then the offline, attack-free, unprotected output $ \textbf{y}_{\U}(k) $ can be recovered as $ \textbf{y}_{\U}(k) = F(\tau_{\rm sw}) \textbf{z}(k) $ for all $ k \in [T_\p, N-1] $.
\end{corollary}
\begin{proof}
	The offline experiment is another trajectory of the same LTI plant.
	Since the protected data are attack-free, $ \widetilde{\textbf{y}}_\cS=\textbf{y}_\cS $.
	Applying Lemma~\ref{lemma:learning_F} to the offline trajectory yields $\textbf{y}_\U(k)=F(\tau_{\rm sw})\textbf{z}(k)$ for all $k\in [T_\p,N-1]$.
\end{proof}

Define the recovered, offline, partial data by $ \widehat{\textbf{y}}_\U(k) \triangleq F(\tau_{\rm sw}) \textbf{z}(k) $ for $ k \in [T_\p, N-1] $.
Then, we can reconstruct the attack-free offline dataset from time $ T_\p $ to $ N-1 $ as $ \widehat{\textbf{D}} \triangleq (\textbf{u}^{[T_\p,N-1]},\widetilde{\textbf{y}}^{[T_\p, N-1]}_\cS, \widehat{\textbf{y}}^{[T_\p, N-1]}_\U) $ (Line~15 in Algorithm~\ref{algorithm:Secure_DeePC}).

It should be emphasized that the recovered data $ \widehat{\textbf{y}}_\U  $ can be used for the detection and identification of poisoning attacks, i.e., the poisoning attacks can be identified by $ \widehat{\textbf{a}}_\U(k) = \widetilde{\textbf{y}}_\U(k) - \widehat{\textbf{y}}_\U(k) $ for all $ k \in [T_\p,N-1] $ (Line 16 in Algorithm~\ref{algorithm:Secure_DeePC}).
Note that this detection and identification result does not contradict Proposition~\ref{proposition:impossibility_detection} because the online data serve as attack-free side information.

\subsection{$ k \geq \tau_{\rm sw}: $ Original DeePC Using Reconstructed Dataset}
Once the attack-free dataset $ \widehat{\textbf{D}} $ is reconstructed, we can use the original DeePC (\ref{eq:DeePC}) based on the dataset while considering the reference signals and constraints for all input and output channels (Line~19 in Algorithm~\ref{algorithm:Secure_DeePC}).
Consequently, we have the following theorem showing the finite-time recovery and MPC-equivalent operation of Secure DeePC.
\begin{theorem}
	\label{theorem:secure_deepc}
	Suppose that $ (\bar A, \bar B) $ is controllable, $ (\bar A, \bar C_{\cS}) $ is observable, and $ T_\p \geq n $.
	Assume further that the partial offline input $ \textbf{u}^{[T_\p, N-1]} $ is PE of order $ n + T_\p + T_\f $.
	Secure DeePC (Algorithm~\ref{algorithm:Secure_DeePC}) achieves MPC-equivalent performance and constraint satisfaction after a finite transient time almost surely if the algorithm is feasible at the switching time $ \tau_{\rm sw} $ and recursively feasible thereafter.
\end{theorem}
\begin{proof}
	By the dithered input policy and \cite[Lemma 1]{2026ShinoharaTCNS}, the online input becomes PE of order $ 2T_\p $ after finite steps almost surely; namely, $ \tau_{\rm pe} < \infty $ almost surely.
	When $ k = \tau_{\rm sw} $, by Lemma~\ref{lemma:learning_F}, the linear mapping can be constructed by $ F(\tau_{\rm sw}) \triangleq Y_\U(\tau_{\rm sw}) Z(\tau_{\rm sw})^\dagger  $.
	Then, by Corollary~\ref{corollary:recovered_U}, the offline, attack-free output data can be reconstructed as $\widehat{\textbf{y}}_\U(k)=\textbf{y}_\U(k)=F(\tau_{\rm sw})\textbf{z}(k)$ for all $k\in [T_\p,N-1]$.
	Thus, the reconstructed dataset $ \widehat{\textbf{D}} $ is equivalent to the attack-free dataset $ (\textbf{u}^{[T_\p,N-1]},{\textbf{y}}^{[T_\p, N-1]}_\cS, {\textbf{y}}^{[T_\p, N-1]}_\U) $ for the interval $ [T_\p, N-1] $.
	
	For Phase II of Secure DeePC (i.e., $ k \geq \tau_{\rm sw} $), since $ (\bar A, \bar B) $ is controllable, $ (\bar A, \bar C) $ is observable, $ T_\p \geq n $, and $ \textbf{u}^{[T_\p, N-1]} $ is PE of order $ n+T_\p +T_\f $, the optimal control sequence and corresponding system output obtained by this DeePC coincide with those of the MPC problem with prediction horizon $ T_\f $ \cite[Corollary 5.1]{2019ECCDorfler}.
	Provided that the algorithm is feasible at the switching time $ \tau_{\rm sw} $ and recursively feasible thereafter, this implies MPC-equivalent performance and constraint satisfaction of Secure DeePC in Phase II.
\end{proof}

Once the online input becomes sufficiently exciting (which is achievable almost surely in finite time), the attack-free offline data in $ [T_\p, N -1] $ can be recovered exactly from the protected outputs.
The Secure DeePC controller thereafter coincides with the original full-output DeePC and hence with the corresponding MPC under the conditions of Theorem~\ref{theorem:secure_deepc}.

The observability condition on $(\bar A, \bar C_\cS)$ is needed to guarantee performance in Phase I.
If an attack-free dataset is available and the offline input data is PE of order $ n + T_\p $ with $ T_\p \geq n $, one can select an optimal placement of $ \cS $ so that $ (\bar A, \bar C_\cS) $ is observable from the dataset by using the rank-condition of Hankel matrices \cite[Proposition 1]{2022L-CSSWaarde}.

This paper focuses on noise-free systems.
Since the proposed method relies on the standard DeePC technique, the natural extension to noisy systems is to incorporate regularized DeePC formulations, which are developed in \cite{2019ECCDorfler}.
However, in the presence of noise, exact recovery of the linear mapping $ \bar F $ cannot be achieved in general, so another approach is needed to guarantee attack resilience against poisoning attacks.
This is left for future work.

\section{Numerical Examples}
\label{section:Simulation}

%
We illustrate the attack-resilience performance of Secure DeePC by comparing it with the regularized DeePC \cite{2019ECCDorfler} and output-truncated DeePC (\ref{eq:Secure_temp_DeePC}) via numerical simulations. 
We use an inverted pendulum system depicted in Fig.~\ref{fig:pendulum}.

\subsection{Settings}

\begin{figure}[t]
	\begin{center}
		\includegraphics[width=0.45\linewidth]{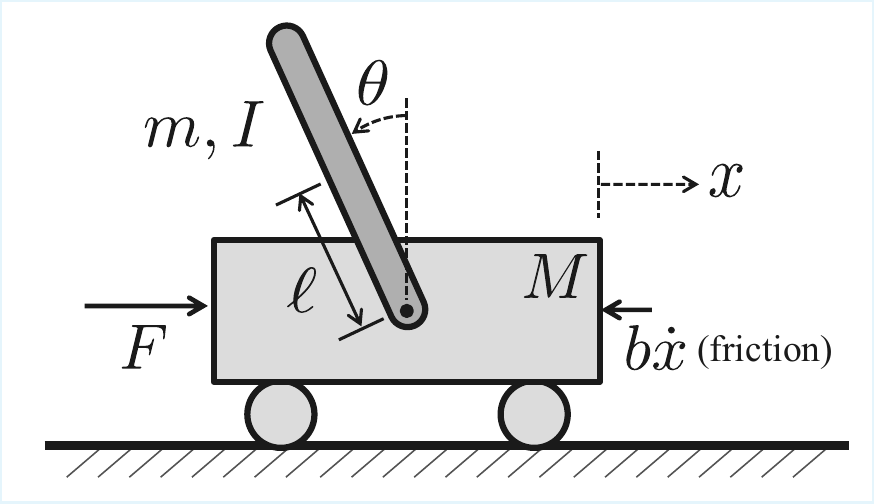}
		\vspace{-2.5mm}
		\caption{Inverted pendulum.}
		\label{fig:pendulum}
		\vspace{-6.5mm}
	\end{center}
\end{figure}

We adopt a linearized model.
The physical parameters are chosen as follows: $ M = 0.5~\mathrm{kg} $, $ m = 0.2~\mathrm{kg} $, $ \ell = 0.3~\mathrm{m} $, $ I = 0.006~\mathrm{kg\cdot m^2} $, and $ b = 0.1~\mathrm{N\cdot s/m} $.
The states of the system are given as $ (x, \dot{x}, \theta, \dot{\theta}) $ and the input is given as $ u = F $.
Assume that the sensors measure $ (x, \theta, \dot{x}). $
Then, the discretized system with the sampling time $ T_s=0.1~\mathrm{s}$ is characterized by (\ref{eq:system_model}).  

We collect $ N = 120 $ data samples during an offline experiment.
The offline input sequence $ \textbf{u} $ is generated randomly, which is PE of sufficiently high order.
For the output data, we assume that only the first two output channels are protected, i.e., $ \cS = \{1,2\} $ and $ \mathcal{U}=\{3\} $.
Referring (\ref{eq:attack_proposition}), we design the poisoning attack sequence as $ \textbf{a}_3(k) = -\alpha \textbf{y}_1(k) $ for all $ k \in[0,119] $ with $ \alpha = 1 $.

The regularized DeePC uses the full dataset $ \textbf{D} $, including the poisoned data.
The output-truncated DeePC uses only the protected dataset $ \textbf{D}_\cS $.
For the objective function $ J(u^\f, y^\f) $, we set $ Q = \mathrm{diag}(10^2,10^2,10^2) $ and $ R = 1 $.
We impose the following constraints:
\begin{align*}
	\mathcal{C}_u \!=\! \left\lbrace u\!:\! -5 \!\leq\! u \!\leq\! 5\right\rbrace\!,~\mathcal{C}_y \!=\! \left\lbrace \!\!\left[\!\begin{smallmatrix}
		y_1 \\ y_2 \\ y_3
	\end{smallmatrix}\!\right]: \left[\!\begin{smallmatrix}
	-4 \\ -\frac{\pi}{6} \\ -1
	\end{smallmatrix}\!\right] \leq \left[\!\begin{smallmatrix}
	y_1 \\ y_2 \\ y_3
\end{smallmatrix}\!\right] \leq  \left[\!\begin{smallmatrix}
4 \\ \frac{\pi}{6} \\ 1
\end{smallmatrix}\!\right] \right\rbrace.
\end{align*}
The regularization parameters (for details, see \cite{2019ECCDorfler}) are set as $ \lambda_g = 1 $ and $ \lambda_\sigma= 10^2 $.
%

%

\begin{figure*}[t!]
	\begin{minipage}[b]{0.3\linewidth}
		\includegraphics[keepaspectratio, scale=0.48]{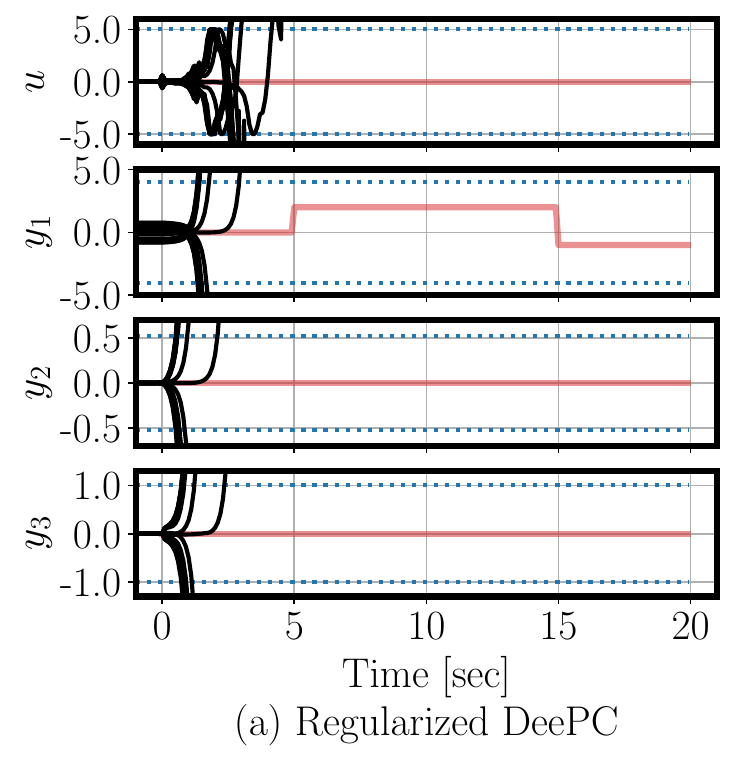}
	\end{minipage}
	\hspace{4mm}
	\begin{minipage}[b]{0.3\linewidth}
		\includegraphics[keepaspectratio, scale=0.48]{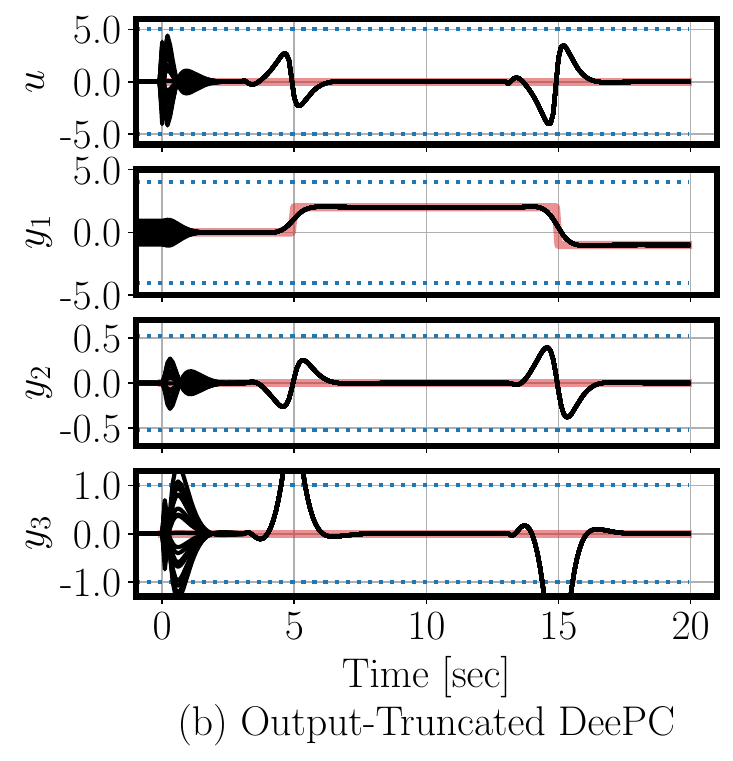}
	\end{minipage}
	\hspace{4mm}
	\begin{minipage}[b]{0.3\linewidth}
		\includegraphics[keepaspectratio, scale=0.48]{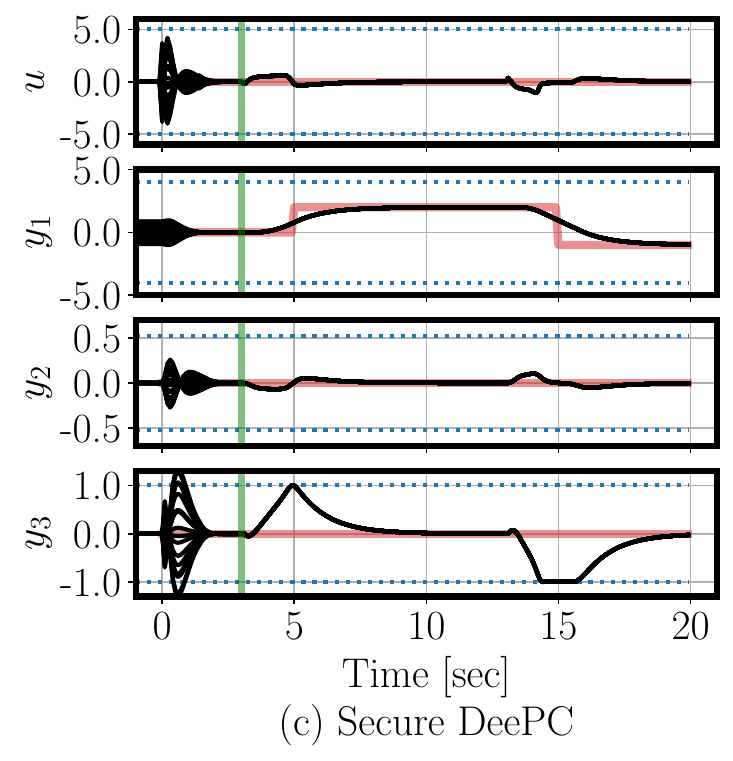}
	\end{minipage}
	\vspace{-2.5mm}
	\caption{I/O trajectories  (solid black) of the regularized DeePC, the output-truncated DeePC, and Secure DeePC under poisoning attacks for 20 trials. The solid red lines indicate the reference signals while the dotted blue lines indicate the constraints.
	The green lines in (c) indicate the first time $ \tau_{\rm pe} $ at which the online input data is PE of order $ 2 T_\p $, which was the same in all 20 trials.
	After one step of $ \tau_{\rm pe} $, Secure DeePC reconstructs the offline attack-free data and thereafter uses the reconstructed  data.}
	\label{fig:simulation_result}
	
	\vspace{-3.5mm}
\end{figure*}

\subsection{Results}
Fig.~\ref{fig:simulation_result} compares the closed-loop I/O trajectories of 20 trials obtained with regularized DeePC, the output-truncated DeePC, and Secure DeePC under offline data poisoning, with $T_\p = 10$ and $T_\f = 20$.
For each trial, the initial position is randomly chosen from $ -1 $ to $ 1 $, and the other initial parameters are set to $ 0 $.
Each DeePC operates from time $ k = 0 $.
In Phase I of Secure DeePC, we add the random dither $ w(k) $, drawn from the uniform distribution on $ [-10^{-4},10^{-4}] $, to the system.

The regularized DeePC (Fig.~\ref{fig:simulation_result}(a)) exhibits a diverging response in all trials, indicating that conventional DeePC can become unstable under data poisoning attacks.
The output-truncated DeePC (Fig.~\ref{fig:simulation_result}(b)) tracks the references well; however, since $ y_3 $ (an output in $ \U $) is not included in the optimization, it violates the safety bound on $ y_3$ (see around $ 5 $ and $ 15 $ seconds in the figure).
Secure DeePC (Fig.~\ref{fig:simulation_result}(c)) shows a small initial violation of the constraint on $ y_3 $ because that constraint is not imposed in Phase I.
However, once the online input data is PE, it tracks the references while respecting all constraints.
Therefore, Secure DeePC achieves MPC-equivalent performance and constraint satisfaction in finite time, demonstrating its resilience to poisoning attacks.

\section{Conclusion}
\label{section:Conclusion}
This paper investigated data-driven control in the presence of data poisoning attacks on partially unprotected outputs.
We first showed fundamental limitations: poisoning attacks cannot be detected and identified from the dataset alone, unprotected data are non-informative in a minimax sense, and hard safety constraints on unprotected outputs are not guaranteed.
Building on these limitations, we proposed Secure DeePC, a two-phase framework that is resilient to offline poisoning attacks.
We showed that Secure DeePC achieves MPC-equivalent guarantees in finite time, almost surely, under certain conditions.
Numerical simulations demonstrated the efficacy of Secure DeePC against offline poisoning attacks.

Future work will pursue secure control problems in the presence of both offline/online attacks.
Another important direction is to address noisy data, including robustness analysis.

\end{document}